\documentclass[12pt]{article}
\usepackage[letterpaper,margin=1in]{geometry}
\usepackage[T1]{fontenc}
\usepackage{lmodern}
\usepackage{amsmath,amssymb,amsthm,mathtools}
\usepackage{microtype}
\usepackage{enumitem}
\usepackage{needspace}
\usepackage{booktabs,tabularx}
\usepackage{float}
\usepackage[round,authoryear]{natbib}
\usepackage{xcolor}
\usepackage[colorlinks=true,linkcolor=blue!45!black,
  citecolor=blue!45!black,urlcolor=blue!45!black,
  pdftitle={Complementary Information Sources},
  pdfauthor={Zichang Wang}]{hyperref}
\setlist[enumerate]{leftmargin=2em,itemsep=0.3em,topsep=0.4em}

\newtheorem{theorem}{Theorem}
\newtheorem{lemma}{Lemma}
\newtheorem{proposition}{Proposition}
\newtheorem{corollary}{Corollary}
\newtheorem{claim}{Claim}
\theoremstyle{definition}
\newtheorem{definition}{Definition}
\newtheorem{example}{Example}
\theoremstyle{remark}

\newcommand{\IR}{\mathbb R}
\newcommand{\E}{\mathbb E}
\newcommand{\ind}{\mathbf 1}
\newcommand{\sB}{\mathcal B}
\newcommand{\sD}{\mathcal D}
\newcommand{\sE}{\mathcal E}
\newcommand{\sM}{\mathcal M}

\newcommand{\coapp}{\operatorname{co}^{\oplus}}
\newcommand{\clco}{\overline{\operatorname{co}}}
\newcommand{\supp}{\operatorname{supp}}
\newcommand{\TV}{\operatorname{TV}}

\title{Complementary Information Sources\footnote{I thank Jian Li, Xiao Lin, and audiences at ``Theory@Chapel Hill 2025'' for feedback and suggestions. All errors are my own.}}
\author{Zichang Wang\thanks{Wang Yanan Institute for Studies in Economics
and the School of Economics, Xiamen University.
Email: \href{mailto:zichang.wang262@gmail.com}{zichang.wang262@gmail.com}.}}
\date{September 28, 2026}

\begin{document}
\hypersetup{pageanchor=false}
\begin{titlepage}
\maketitle
\thispagestyle{empty}
\begin{abstract}
A decision maker may have several information sources available and choose which one to consult only after learning the decision problem she faces.  When is one such set of sources uniformly more valuable than another?  For unrestricted Bayesian decision problems, we show that the answer can be stated entirely in terms of Blackwell comparisons.  Form a tagged mixture by drawing a source independently of the state and revealing both its identity and its signal.  One source set is more valuable in every decision problem if and only if each tagged mixture of the second source set is Blackwell dominated by some tagged mixture of the first. The result applies to compact, possibly infinite source sets and general signal spaces. It also has an exact quantitative counterpart: the largest normalized value shortfall is the directed Le Cam deficiency between the source sets' tagged hulls. The analogous program for monotone decision problems reveals a boundary. We call the passage from problem-by-problem source-set superiority to a
problem-independent pairwise dominance a \emph{lifting}. The Blackwell lifting does not extend directly to the Lehmann order: mixing sources that individually satisfy the monotone likelihood ratio property (MLRP) need not preserve MLRP, and even when it does, no fixed Lehmann-dominating mixture need exist.  Requiring one source to serve a finite bundle of monotone decision problems restores the equivalence.
\end{abstract}
\vspace{0.8em}

\end{titlepage}
\pagenumbering{arabic}
\hypersetup{pageanchor=true}

\section{Introduction}
\label{sec:intro}

Different information sources are useful in different decision problems.  A credit bureau may be best for assessing default, an industry report for forecasting demand, and a technical audit for evaluating production risk.  A firm with access to all three need not combine their signals.  It may wait until the decision problem is known and then consult the source suited to that problem.  The same timing appears when a physician chooses a diagnostic test after the clinical question is clear, or when a policymaker selects a forecasting model after identifying the particular policy question the model must inform.

This paper studies the value of access to such a \emph{source set}. The decision maker knows her prior, actions, and payoffs before selecting one source from the source set.  She then observes that source's signal and acts.  A source set $P$ is more valuable than a source set $Q$ if it gives at least as high an expected payoff in every Bayesian decision problem. This is an upper-envelope comparison: the source chosen from $P$ may depend on the problem.  It is not a comparison of joint signals, so no correlation among the sources in a source set is needed.

The central question is whether this behavioral comparison has a statistical counterpart.  Blackwell's theorem answers the pairwise version: one source is more valuable than another in every decision problem exactly when the latter can be obtained by garbling the former.  With source sets, however, the quantifiers are different.  For every problem there may be a different useful source in $P$.  A statistical comparison, by contrast, must identify sources before the problem varies.  It is not immediate that the first statement contains enough uniformity to imply the second.

The relevant operation is simple.  Draw a source from a state-independent lottery and reveal the source's identity together with its signal.  We call the result a \emph{tagged mixture}.  The tag matters: it lets the decision maker interpret the signal using the experiment that generated it.  Expected value under a tagged mixture is therefore the weighted average of the component values.  The set of all tagged mixtures of $P$, denoted $\coapp(P)$, is its \emph{tagged hull}.

Our first main result is a source-set version of Blackwell's theorem.  A compact source set $P$ is more valuable than a compact source set $Q$ in every decision problem if and only if
\begin{equation*}
 \text{for every }\widetilde q\in\coapp(Q),\quad \text{some }\widetilde p\in\coapp(P) \text{ Blackwell dominates }\widetilde q.
 \tag{B}
\end{equation*}
The mixture $\widetilde p$ may depend on $\widetilde q$, but not on the decision problem.  For the leading case $P=\{p_1,p_2\}$ and $Q=\{q\}$, condition (B) says that a single weight $t$ exists such that the tagged mixture $t p_1\oplus(1-t)p_2$ Blackwell dominates $q$.  We call $p_1$ and $p_2$ \emph{complementary relative to $q$} when their source set dominates $q$ in this sense.  Neither source need dominate $q$ on its own.

A three-state example gives the basic idea.  Let $p_1$ reveal whether the state is the first state and let $p_2$ reveal whether it is the third.  Let $q$ choose each source with probability one half and disclose which source was chosen.  The source set $\{p_1,p_2\}$ is at least as valuable as $q$: after the problem is known, choosing the better of two sources yields at least their average value.  Yet $p_1$ is strictly worse than $q$ for a decision problem that rewards identifying the third state, while $p_2$ is strictly worse for the analogous problem about the first state.  The two sources are valuable as a portfolio of specializations.  The theorem's content is the converse: every pair that is uniformly complementary must admit exactly this kind of statistical certificate, for some lottery weight.

There is a useful geometric way to read the result.  An experiment determines the set of state-contingent action distributions that its decision rules can generate.  Expected utility is linear on this feasible set.  Choosing from a source set takes a union of feasible sets; testing every decision problem compares the closed convex hulls of those unions.  A tagged source lottery implements precisely that convexification.  For each fixed benchmark mixture, this yields a source lottery that works across all decision problems.  The argument extends the feasible-set perspective used in the proof of Blackwell's theorem by \citet{deOliveira2018}.

The same logic suggests a general program: begin with a pairwise information order and ask whether it \emph{lifts} to source sets.  The Blackwell order gives a positive answer, but the Lehmann order exposes what the argument needs.  We study decision problems with ordered states and scalar single-crossing payoffs. The pairwise comparison builds on \citet{Lehmann1988} and later work by \citet{Persico2000}, \citet{AtheyLevin2018}, and \citet{QuahStrulovici2009}. \citet{Kim2023} studies a broader class that permits multidimensional actions. For individual sources in the regular MLRP domain studied below, the Lehmann order is equivalent to higher value in every problem in our class.  Two difficulties appear at the source-set level.

First, a tagged mixture of sources satisfying MLRP need not itself satisfy MLRP. Signals from different sources can induce posterior beliefs that cannot be placed on one MLR chain.  We therefore identify an admissible domain in which all sources in the same source set have posterior supports on a common chain.  This restriction is economically transparent: although sources may differ in precision, they rank the possible signal realizations in a compatible way.  It also makes every tagged mixture satisfy MLRP.

Second, admissibility is not enough.  We construct a four-state counterexample in which, for every single monotone problem, one of two sources outperforms a benchmark, but no fixed lottery over the two sources outperforms the benchmark in all monotone problems. The counterexample remains valid if all decision problems are required to share the same uniform full-support prior. Variation in the prior is therefore not responsible for the failure. Indeed, no alternative pairwise relation between experiments can produce a universal lifting theorem on this domain.  The failure is not statistical; it comes from the decision class.  Two monotone problems need not be represented as one scalar-action monotone problem, so the value profiles generated by primitive problems need not be convex.

This negative result has an economic counterpart.  Choosing a source anew for each decision problem is different from choosing one source that must inform several decision problems. We therefore consider finite bundles of monotone problems.  The source is selected after the bundle is known, but it must serve every component; actions can still be tailored to their respective decision problems.  On this class, the source set comparison lifts exactly to the Lehmann order.  Bundling is thus both the mathematical convexification that restores the theorem and a substantive restriction on how an information source will be used.

Our final result makes the comparison quantitative.  Normalize utility to $[0,1]$ and ask for the largest payoff by which $Q$ can outperform $P$.  This worst-case shortfall equals the directed Le Cam deficiency from the tagged hull of $P$ to the tagged hull of $Q$.  Equivalently, it is the smallest $\varepsilon$ such that every tagged mixture from $Q$ can be reproduced from some tagged mixture of $P$ with state-by-state total-variation error at most $\varepsilon$.  The formula gives the behavioral loss and the statistical simulation error in the same units.

Table~\ref{tab:overview} summarizes the boundary of the lifting argument.
\begin{table}[H]
\centering
\caption{Pairwise information orders and source set comparisons}
\label{tab:overview}
\small
\renewcommand{\arraystretch}{1.16}
\begin{tabularx}{\textwidth}{@{}p{0.26\textwidth}p{0.22\textwidth}X@{}}
\toprule
Decision criterion & Pairwise comparison & Source set implication \\
\midrule
All Bayesian problems & Blackwell order & Exact lifting to tagged hulls \\
One monotone problem & Lehmann order & Lifting fails, even on the admissible domain \\
Finite bundles of monotone problems & Lehmann order & Exact lifting on the admissible domain \\
All $[0,1]$-valued problems & Le Cam deficiency & Exact formula for the maximal value shortfall \\
\bottomrule
\end{tabularx}
\end{table}
\paragraph{Relation to the Literature.} The paper builds on the equivalence among garbling, feasible decision rules, and universal value comparisons in \citet{Blackwell1951,Blackwell1953}; the proof strategy is closest in spirit to \citet{deOliveira2018}. \citet{Jakobsen2021} uses the same disclosed-mixture operation to formulate axioms for preferences over experiments. We extend this construction to arbitrary lotteries over compact source sets and general signal spaces, and use it as the convexification behind the source-set lifting results. For finite source sets, the Blackwell result is also obtainable from the general dominance--optimality analysis of \citet{ChengBorgers2026}: their result that an experiment not dominated by mixtures of the others is uniquely optimal in some decision problem yields the finite complementarity statement by contraposition. We use this observation as a starting point. Our focus is on the source set comparison itself, its extension to compact and possibly infinite families with general signal spaces, and what survives under ordered and approximate notions of informativeness.

Source sets are related to menus of experiments, but the economic question here is different. \citet{BergemannBonattiSmolin2018} study a seller who designs and prices a menu of experiments for a data buyer. In our model, the source sets are given, there are no prices, and the objective is to compare two sets across decision problems. Concurrent work by \citet{Song2026} also studies a menu that is fixed before a context is revealed. That paper considers repeated sampling. We instead study an exact one-shot comparison: after the decision problem is known, the decision maker selects one source and observes it once.

The distinction between choosing and aggregating sources is also important. \citet{deOliveiraIshiiLin2026} study a decision maker who observes several signals whose marginal experiments are known but whose correlation is not. For two states and two actions, their robust strategy uses only the best source; with more actions, it can assign different local comparisons to different specialist sources. \citet{BorgersHernandoKrahmer2013} study informational complements through the incremental value of jointly observed signals. \citet{LiangMu2020} study sequential acquisition from several sources. Although each researcher chooses one source at a time in their model, observations accumulate across periods. A complementary set is one whose sources are jointly needed for long-run identification. In our model, signals from different sources are never combined, either at one time or across time. Our use of ``complementary'' is relative to a benchmark and refers to the value of access to alternative sources across decision problems.

Sets of experiments have a different interpretation in \citet{Wang2024}, where they represent ambiguity about the information-generating mechanism and are evaluated pessimistically. A source set here contains known sources from which the decision maker chooses, and is therefore evaluated by an upper envelope. The ordered part is related to the literature on information in monotone decision problems. \citet{Lehmann1988} introduces the accuracy order, \citet{Persico2000} applies it to single-crossing problems, and \citet{AtheyLevin2018} and \citet{QuahStrulovici2009} provide broader analyses of monotone decision problems. The framework of \citet{AtheyLevin2018} can hold the prior fixed while payoffs vary. Our negative lifting result is not caused by prior variation: the two separating problems in the four-state example share the uniform prior. \citet{LiZhou2020} show that the Lehmann comparison is robust to uncertainty-averse preferences. \citet{Kim2023} allows multidimensional actions and characterizes pairwise comparisons through monotone quasi-garbling. The monotone-information results just cited concern pairwise comparisons. Our question is whether a pairwise order can be lifted to source sets, and which closure of the decision class is needed for such a result.

The approximate part uses the randomization criterion of \citet{LeCam1964}. A general treatment of comparison, randomization, and deficiency is given by \citet{Torgersen1991}. We apply this pairwise theory after taking tagged hulls and obtain an exact formula for the worst source-set value shortfall.

Section~\ref{sec:setup} introduces the model. Section~\ref{sec:blackwell} gives the Blackwell results. Section~\ref{sec:lehmann} studies monotone decision problems, first establishing the failure for primitive decision problems and then the positive theorem for decision problem bundles. Section~\ref{sec:lecam} develops the exact deficiency formula. Section~\ref{sec:conclusion} concludes.  The appendices contain technical details and a separate treatment of Lehmann's original loss formulation.

\section{The Model}
\label{sec:setup}
\noindent\textbf{Sources, decision problems, and garblings.} There are three stages.  First, the decision maker learns the decision problem---her prior, feasible actions, and payoffs.  Second, she selects one source from the source set available to her and observes its signal.  Third, she chooses an action.  Source selection cannot depend on the realized state. This timing makes a source set valuable because it offers specialization across decision problems, not because the signals in the source set can be combined.

Let $\Omega=\{\omega_1,\ldots,\omega_m\}$ be a finite state space, where $m\geq 2$. A \emph{standard Borel space} is a measurable space measurably isomorphic to a Borel subset of a complete separable metric space. For such a space $S$, write $\sB(S)$ for its sigma-algebra and $\Delta(S)$ for the probability measures on it. Finite and countable spaces, Euclidean spaces, and compact metric spaces are standard Borel.

An \emph{information source} is an experiment $p:\Omega\to\Delta(S)$ on a standard Borel signal space $S$. Conditional on state $\omega$, its signal has distribution $p(\cdot\mid\omega)$. Signal spaces may differ across sources.

Given two standard Borel spaces $S$ and $T$, a \emph{(Markov) kernel} is a mapping $K:S\to\Delta(T)$. It assigns a probability measure $K(\cdot\mid s)$ to each $s\in S$, with $s\mapsto K(B\mid s)$ measurable for every measurable $B\subseteq T$. Kernels describe randomized transformations of observations. Given kernels $K:S\to\Delta(T)$ and $L:T\to\Delta(R)$, their composition is the kernel $L\circ K:S\to\Delta(R)$ defined by
\[
 (L\circ K)(C\mid s)=\int_T L(C\mid t)K(dt\mid s).
\]
Let $q:\Omega\to\Delta(T)$ be another source. We write $p\trianglerighteq q$ if $q$ is a \emph{garbling} of $p$: there is a Markov kernel $K:S\to\Delta(T)$ such that
\[
 q(B\mid\omega)=\int_S K(B\mid s)p(ds\mid\omega)\quad\text{for all }\omega\in\Omega,\ B\in\sB(T).
\]
This is the Blackwell order. Sources are \emph{Blackwell equivalent} if each is a garbling of the other.

To avoid set-theoretic bookkeeping over all possible action spaces, fix the Hilbert cube $\mathcal A=[0,1]^{\mathbb N}$ as a grand action space. Every compact metric space is homeomorphic to a compact subset of $\mathcal A$; we therefore identify action spaces with such subsets. We also fix an order-preserving copy of $\IR$ in $\mathcal A$ for the monotone problems in Section~\ref{sec:lehmann}.

A \emph{decision problem} is $D=(A,u,\pi)$, where $A$ is a nonempty compact subset of $\mathcal A$, $u:\Omega\times A\to\IR$ is continuous in its action argument for each state, and $\pi\in\Delta(\Omega)$ is a prior. Let $\sD$ denote the set of these problems. In particular, finite action spaces are allowed. An \emph{action plan} for $p$ is a Markov kernel $\alpha:S\to\Delta(A)$. The source's ex-ante value is
\[
 U(p,D)=\sup_\alpha\sum_{\omega\in\Omega}\pi(\omega)\int_S\int_A u(\omega,a)\alpha(da\mid s)p(ds\mid\omega).
\]
Utilities are bounded because $\Omega$ is finite and $A$ is compact. An optimal measurable deterministic plan exists: conditional expected utility is measurable in the signal and continuous in the action, so a measurable maximizer can be selected on a compact metric action space. Randomized plans are retained to describe feasible state-dependent action distributions. For a collection $P$ of sources, write
\[
 U(P,D)=\sup_{p\in P}U(p,D).
\]
The supremum is a maximum under the compactness assumption below.  Thus the source set is evaluated after the decision problem is known.  For a class $\mathcal C$ of decision problems, write
\[
 P\succeq_{\mathcal C}Q\quad\Longleftrightarrow\quad U(P,D)\geq U(Q,D)\quad\text{for every }D\in\mathcal C.
\]
The unrestricted comparison uses $\mathcal C=\sD$.  Later sections vary $\mathcal C$ while keeping the timing of source selection fixed.

\subsection{Canonical Sources and Compact Collections}

Source sets may contain sources with different signal spaces, so compactness cannot be imposed by placing all probability matrices in one Euclidean space.  We instead identify Blackwell-equivalent sources with the same distribution of posterior beliefs.  This gives a common compact space of experiments and makes clear that the results do not depend on how signals are labeled.  Readers interested only in finite sources and finite signal spaces may view this subsection as the usual posterior representation of an experiment.

Fix the full-support reference prior $\bar\pi(\omega)=1/m$ and let $X=\Delta(\Omega)$.  Under this prior, a source $p$ induces a distribution $\tau_p$ over posterior beliefs in $X$.  Its mean is $\bar\pi$.  Conversely, any probability measure $\tau$ on $X$ with that mean defines the canonical
experiment
\[
 \widehat p_\tau(B\mid\omega)=\frac{1}{\bar\pi(\omega)}\int_B\gamma(\omega)\tau(d\gamma).
\]
This formula applies to every $B\in\sB(X)$.

\begin{lemma}[Canonical representation]
\label{lem:canonical}
Every source $p$ is Blackwell equivalent to $\widehat p_{\tau_p}$. Two sources are Blackwell equivalent if and only if their reference posterior distributions coincide.
\end{lemma}
The construction and proof, including the measure-theoretic details for standard Borel signals, are given in Appendix~\ref{app:canonical}.

Accordingly, the space of equivalence classes can be identified with
\[
 \sE=\left\{\tau\in\Delta(X)\mid\int_X\gamma\,\tau(d\gamma)=\bar\pi\right\}.
\]
We equip this space with \emph{weak convergence}: $\tau_n\to\tau$ means $\int f\,d\tau_n\to\int f\,d\tau$ for every continuous real function $f$ on $X$. Probability measures on a compact metric space form a compact metrizable space in this topology. The mean restriction displayed above is closed, so $\sE$ is compact metric. We continue to write $p$ for a source or its equivalence class, and write $\tau_p$ when integrating over beliefs. A \emph{source set} is a nonempty compact subset of $\sE$. Compactness guarantees that an optimal source is available rather than merely approached by a sequence of sources.

For any decision problem $D=(A,u,\pi)\in\sD$, define
\[
 \phi_D(\gamma)=\max_{a\in A}\sum_{\omega\in\Omega}\frac{\pi(\omega)}{\bar\pi(\omega)}\gamma(\omega)u(\omega,a),
\]
then the value of source $p$ for this problem $D$ can be written as
\begin{equation}
 U(p,D)=\int_X\phi_D(\gamma)\tau_p(d\gamma).
 \label{eq:posteriorvalue}
\end{equation}
The value formula follows by optimizing after each canonical signal. The function $\phi_D$ is continuous and convex: it is a maximum of linear functions whose coefficients range over a compact set. Consequently $p\mapsto U(p,D)$ is continuous on $\sE$, and the supremum defining $U(P,D)$ is attained whenever $P$ is nonempty and compact.

\subsection{Tagged Mixtures}

Let $P\subseteq\sE$ be nonempty and compact. A source lottery $\xi\in\Delta(P)$ is a Borel probability distribution over $P$, independent of the state. The associated \emph{tagged mixture} $p_\xi$ has signal space $P\times X$ and law
\[
 p_\xi(C\mid\omega)=\int_P\int_X\ind_C(r,\gamma)\widehat p_{\tau_r}(d\gamma\mid\omega)\xi(dr).
\]
The source label $r$ is observed. For finite source sets, this is the matrix-appending construction used by \citet{Jakobsen2021}; our formulation allows arbitrary lotteries over a compact source set and general signal spaces. This is not the ordinary mixture that first forgets which probability law generated the signal. The tag lets the decision maker use a source-specific action rule and is exactly what makes value affine in the lottery. The integrand is measurable: integration of a bounded Borel function against a varying probability measure is Borel measurable, and the canonical kernel is given by the formula above. We use $p_\xi$ for both the tagged experiment and its Blackwell-equivalence class in $\sE$, represented canonically by its posterior distribution $\tau_{p_\xi}$. Define the \emph{tagged hull}
\[
 \coapp(P)=\{p_\xi\mid\xi\in\Delta(P)\}\subseteq\sE.
\]
Thus the elements of $\coapp(P)$ are Blackwell-equivalence classes of experiments---equivalently, their canonical posterior distributions. For $P=\{p_1,p_2\}$, write $p_t=t p_1\oplus(1-t)p_2$. For finite signal spaces, tagged mixing concatenates the probability matrices after multiplying each block by its lottery weight.

\begin{lemma}[Value and convexity of tagged mixtures]
\label{lem:append}
For every $D\in\sD$ and $\xi\in\Delta(P)$, the posterior distribution of
the mixture is $\tau_{p_\xi}=\int_P\tau_r\,\xi(dr)$. 
Its value is
\[
 U(p_\xi,D)=\int_P U(r,D)\xi(dr).
\]
The first identity means equality of the measures on every Borel subset of $X$. The set $\coapp(P)$ is compact and convex and equals the weakly closed convex hull of $P$ in $\sE$. In particular, $U(\coapp(P),D)=U(P,D)$.
\end{lemma}

\begin{proof}
Under $\bar\pi$, the joint law of $(r,\gamma)$ is $\xi(dr)\tau_r(d\gamma)$, and its posterior is $\gamma$. This proves the first identity. The second follows from \eqref{eq:posteriorvalue} and iterated integration. For every continuous $f$ on $X$,
\[
 \int_X f\,d\tau_{p_\xi}=\int_P\left(\int_X f\,d\tau_r\right)\xi(dr).
\]
The inner integral is continuous in $r$, so the map $\xi\mapsto\tau_{p_\xi}$ is continuous and affine. Its image is compact and convex because $\Delta(P)$ is compact and convex. Finitely supported probabilities are weakly dense in $\Delta(P)$: partition the compact metric space into finitely many Borel sets of arbitrarily small diameter and move each set's mass to one representative. Their images are precisely the finite convex combinations of members of $P$. This proves the hull identity. The value identity gives the last assertion.
\end{proof}

\subsection{Feasible Action Distributions}

Following a construction similar to that in \citet{deOliveira2018}, for a compact action space $A$, the distributions of actions that can be implemented state by state are
\[
 \Lambda_p(A)=\left\{\lambda\in\Delta(A)^\Omega\mid\lambda(B\mid\omega)=\int_S\alpha(B\mid s)p(ds\mid\omega)\text{ for some action plan }\alpha\right\}.
\]
Set $\Lambda_P(A)=\bigcup_{p\in P}\Lambda_p(A)$. The superscript $\Omega$ denotes one probability measure for each state. We use the product weak topology on $\Delta(A)^\Omega$. The notation $\clco$ denotes the closed convex hull in that topology: the smallest closed convex set containing the specified set.

\subsection{Finite Approximations and Bundles of Problems}

Two elementary properties of decision problems are central to the analysis. We first define the operation used in the second one. Given $D_k=(A_k,u_k,\pi_k)\in\sD$, $k=1,\ldots,n$, and $\theta\in\Delta(\{1,\ldots,n\})$, let $A^*=\prod_k A_k$, identified with a compact subset of $\mathcal A$, and let $\pi^*=\sum_k\theta_k\pi_k$. When $\pi^*(\omega)>0$, define
\[
 u^*(\omega,(a_1,\ldots,a_n))=\sum_k\frac{\theta_k\pi_k(\omega)}{\pi^*(\omega)}u_k(\omega,a_k).
\]
Set $u^*(\omega,\cdot)=0$ when $\pi^*(\omega)=0$. We call the resulting problem $D^*=(A^*,u^*,\pi^*)$ the \emph{weighted direct sum} and write $D^*=\bigoplus_k\theta_kD_k$.

\begin{lemma}[Finite approximation and direct sums]
\label{lem:decisions}
\begin{enumerate}[label=(\roman*)]
\item For every $D=(A,u,\pi)\in\sD$ and $h>0$, there is a finite
      $A_h\subseteq A$ such that, for $D_h=(A_h,u|_{\Omega\times A_h},\pi)$,
      \begin{equation}
       0\leq U(p,D)-U(p,D_h)\leq h
       \quad\text{for every source }p.
       \label{eq:finiteapprox}
      \end{equation}
\item Given $D_k=(A_k,u_k,\pi_k)\in\sD$, $k=1,\ldots,n$, and
      $\theta\in\Delta(\{1,\ldots,n\})$, the weighted direct sum
      $D^*=\bigoplus_k\theta_k D_k$ belongs to $\sD$ and satisfies
      \begin{equation}
       U(p,D^*)=\sum_k\theta_k U(p,D_k)
       \quad\text{for every }p.
       \label{eq:directsum}
      \end{equation}
      If all the $u_k$ take values in $[0,1]$, so does the utility
      in $D^*$.
\end{enumerate}
\end{lemma}

\begin{proof}
Consider the continuous map that assigns to each action its vector of state-contingent utilities, $a\mapsto(u(\omega,a))_{\omega\in\Omega}\in\IR^m$. Its image is compact. It can therefore be covered by finitely many balls of radius $h$, centered at points of the image, in the maximum norm, where the distance between two vectors is the largest absolute coordinate difference. Retain an action corresponding to each center and collect these actions in $A_h$. For every $a\in A$, some $a_h\in A_h$ then satisfies $|u(\omega,a)-u(\omega,a_h)|\leq h$ in every state $\omega$. Optimizing in \eqref{eq:posteriorvalue} and integrating proves \eqref{eq:finiteapprox}; the integrated state weights sum to one.

For the second assertion, use the weighted direct sum defined before the lemma. The product action space is compact metric and the utility is continuous. Conditional maximization separates across the coordinates $a_k$, which proves \eqref{eq:directsum}. In every state with positive probability, $u^*$ is a convex combination of the component utilities, proving the normalization claim.
\end{proof}

Part~(ii) has a direct interpretation.  The product action chooses one action for each component problem, and the same signal can be used in every component.  Thus $D^*$ is a weighted bundle of decision problems to be served by one source.  Unrestricted Bayesian problems are closed under this operation.  The closure looks innocuous here, but it will be exactly what distinguishes the positive and negative results for monotone problems.

It follows from the uniform bound \eqref{eq:finiteapprox} that testing all finite-action problems is equivalent to testing all of $\sD$ for comparisons of the form $U(P,D)\geq U(Q,D)-c$, with a fixed tolerance $c$. The same approximation applies to a restricted decision class whenever restricting to a finite subset of actions preserves its utility restrictions.

\section{Blackwell Comparisons}
\label{sec:blackwell}
We start with preliminaries about Blackwell comparisons of two information sources.
\begin{lemma}[Pairwise Blackwell comparison]
\label{lem:blackwell}
For any sources $p,q$, the following are equivalent:
\begin{enumerate}[label=(\roman*)]
\item $p\trianglerighteq q$.
\item $\Lambda_p(A)\supseteq\Lambda_q(A)$ for every nonempty compact
      metric action space $A$.
\item $U(p,D)\geq U(q,D)$ for every $D\in\sD$.
\end{enumerate}
\end{lemma}

\begin{proof}
The finite version is Theorem~1 of \citet{deOliveira2018}.

If $q=K\circ p$, a plan $\alpha$ for $q$ can be implemented under $p$ as $\alpha\circ K$. Kernel composition is associative by iterated integration. This proves~(i)$\Rightarrow$(ii), while (ii)$\Rightarrow$(iii) follows by maximizing the same linear payoff over the two feasible sets.

For~(iii)$\Rightarrow$(i), use canonical signals and set $A=X$. The identity plan under $q$ implements the vector $q=(q(\cdot\mid\omega))_{\omega\in\Omega}$. If $p$ does not garble into $q$, this vector is outside $\Lambda_p(X)$. This set is compact and convex: it is the continuous affine image of the compact convex set of joint laws $\zeta\in\Delta(X\times X)$ with first marginal $\tau_p$, under $\zeta\mapsto\bigl(\bar\pi(\omega)^{-1}\int_{X\times\cdot}\gamma(\omega)\,\zeta(d\gamma,da)\bigr)_{\omega\in\Omega}$. Separation in the weak topology gives continuous functions $v_\omega:X\to\IR$ such that
\[
 \sum_\omega\int_X v_\omega(a)q(da\mid\omega)
 >
 \sup_{\lambda\in\Lambda_p(X)}\sum_\omega\int_X v_\omega(a)\lambda(da\mid\omega).
\]
Here a weakly continuous linear functional is a finite sum of integrals against continuous functions, which explains the form of the separator. For the problem with prior $\bar\pi$ and utility $u(\omega,a)=v_\omega(a)/\bar\pi(\omega)$, the left side is the payoff of the identity plan under $q$, and the right side is $U(p,D)$. Hence $U(q,D)>U(p,D)$. This proves~(iii)$\Rightarrow$(i) and also shows that feasible-plan inclusion at $A=X$ suffices for~(i). Finally, Lemma~\ref{lem:decisions}(i) replaces a strict compact-action value gap by a strict finite-action gap.
\end{proof}

\subsection{Two Sources Relative to One Benchmark}

We say that $p_1$ and $p_2$ are \emph{complementary relative to $q$} if
\[
 \max\{U(p_1,D),U(p_2,D)\}\geq U(q,D)\quad\text{for every }D\in\sD.
\]
This is a weak comparison. In particular, it includes the case in which one source by itself dominates $q$.
\begin{theorem}[Two-source complementarity]
\label{thm:pair}
For any three information sources $p_1,p_2,q$, the following statements are equivalent:
\begin{enumerate}[label=(\roman*)]
\item $p_1$ and $p_2$ are complementary relative to $q$.
\item There exists $t\in[0,1]$ such that $t p_1\oplus(1-t)p_2\trianglerighteq q.$
\end{enumerate}
The same equivalence holds if condition~(i) is tested only on finite-action problems. The lottery weight in~(ii) is independent of the decision problem.
\end{theorem}

The easy direction illustrates why the source identity is disclosed.  If the
tagged mixture $p_t$ dominates $q$, then
\[
 U(q,D)\leq U(p_t,D)=tU(p_1,D)+(1-t)U(p_2,D)\leq\max\{U(p_1,D),U(p_2,D)\}.
\]
The content of the theorem is the reverse direction.  Although the source attaining the maximum may vary with $D$, one lottery weight works for all problems.  We prove this after establishing the decision-rule representation and the compactness argument that will also be used for comparisons of source sets.

\begin{lemma}[Geometry of action plans]
\label{lem:geometry}
If $R\subseteq\sE$ is nonempty and compact and $A$ is compact metric, then $\Lambda_R(A)$ is compact. It is convex if $R$ is convex. For every nonempty compact $P\subseteq\sE$,
\begin{equation}
 \Lambda_{\coapp(P)}(A)=\clco\bigl(\Lambda_P(A)\bigr).
 \label{eq:lambda-hull}
\end{equation}
\end{lemma}

\begin{proof}
Let $\zeta\in\Delta(X\times A)$ have first marginal $\tau_r$ for some $r\in R$. The set of all such $\zeta$ is compact: the space $\Delta(X\times A)$ is compact and taking a marginal is continuous. It is convex if $R$ is convex. Define
\[
 (T\zeta)(B\mid\omega)=\frac{1}{\bar\pi(\omega)}\int_{X\times B}\gamma(\omega)\zeta(d\gamma,da).
\]
This is a probability measure in each state. The map $T$ is continuous in the weak topologies, because for every continuous $v:A\to\IR$ its integral against $(T\zeta)(\cdot\mid\omega)$ is the integral of the continuous function $\gamma(\omega)v(a)/\bar\pi(\omega)$ against $\zeta$. It is also affine.

Disintegrating $\zeta$ given $\gamma$ gives an action plan $\alpha(da\mid\gamma)$ with $\zeta(d\gamma,da)=\tau_r(d\gamma)\alpha(da\mid\gamma)$. The definition of $T$ then gives $T\zeta=\alpha\circ r$. Conversely, every action plan generates such a joint law. Thus the image of the compact set under $T$ is exactly $\Lambda_R(A)$, proving the first two assertions.

Apply these assertions to $R=\coapp(P)$. Its image is compact and convex and contains $\Lambda_P(A)$, so it contains the right side of \eqref{eq:lambda-hull}. For the reverse inclusion, take $p_\xi\in\coapp(P)$ and represent any of its feasible action distributions using a plan $\alpha$ on its canonical signals. Since $\tau_{p_\xi}=\int_P\tau_r\,\xi(dr)$, write $\lambda_r=\alpha\circ r$. The distribution generated by $p_\xi$ is the state-by-state average
\[
       \bar\lambda=\alpha\circ p_\xi=\int_P\lambda_r\,\xi(dr).
\]
The map $r\mapsto\lambda_r$ is measurable, and every $\lambda_r$ belongs to $\Lambda_P(A)$. Let $C=\clco(\Lambda_P(A))$. Suppose, contrary to the desired inclusion, that $\bar\lambda\notin C$. Separation supplies continuous functions $v_\omega:A\to\IR$ such that the continuous linear functional
\[
 \ell(\lambda)=\sum_{\omega\in\Omega}\int_A v_\omega(a)\lambda(da\mid\omega)
\]
satisfies $\ell(\bar\lambda)>\sup_{\lambda\in C}\ell(\lambda)$. Yet
\[
 \ell(\bar\lambda)=\int_P\ell(\lambda_r)\xi(dr).
\]
Because every $\lambda_r$ belongs to $C$, the right-hand side is at most $\sup_{\lambda\in C}\ell(\lambda)$, a contradiction. This completes the proof.
\end{proof}

A finite-dimensional separation argument supplies the uniform lottery. It is applied to a compact set of candidate mixtures and requires no topology on the collection of decision problems.

\begin{lemma}[A Uniformization and Minimax Lemma]
\label{lem:minimax}
Let $K$ be a nonempty compact convex subset of a topological vector space. Let $\{f_d\mid d\in I\}$ be a nonempty family of continuous affine real functions on $K$, closed under finite convex combinations. For any real $c$,
\[
 \min_{x\in K}f_d(x)\leq c\text{ for every }d\in I\quad\Longrightarrow\quad\text{some }x\in K\text{ satisfies }f_d(x)\leq c\text{ for every }d\in I.
\]
Moreover, $\min_{x\in K}\sup_{d\in I}f_d(x)=\sup_{d\in I}\min_{x\in K}f_d(x)$ whenever the right side is finite, and the minimum on the left is attained.
\end{lemma}

\begin{proof}
Suppose the closed sets $\{x\mid f_d(x)\leq c\}$ have empty intersection. Compactness supplies a finite subfamily $d_1,\ldots,d_n$ with empty intersection. The compact convex set
\[
 G=\{(f_{d_1}(x),\ldots,f_{d_n}(x))\mid x\in K\}
\]
is disjoint from $(-\infty,c]^n$. Strict separation of a compact convex set and this closed convex set gives nonnegative weights $\theta_1,\ldots,\theta_n$, summing to one, with
\[
       \min_{x\in K}\sum_k\theta_k f_{d_k}(x)>c.
\]
The separating weights must be nonnegative because $(-\infty,c]^n$ is unbounded in each negative coordinate direction. The convex combination is a member of the given family, a contradiction. This proves the implication. Apply it with $c=\sup_d\min_x f_d(x)$ to obtain the upper bound in the displayed minimax identity, with attainment. The reverse inequality always holds.
\end{proof}

Say that a nonempty decision class $\mathcal C\subseteq\sD$ is \emph{closed under finite direct sums} if the problems constructed in Lemma~\ref{lem:decisions}(ii) remain in the class. It is enough that an equivalent problem, with the same value at every source, belongs to the class.

The source-set lifting argument can now be stated entirely in terms of decision values, without invoking a statistical order. Closure under bundles lets us replace a source choice that may vary from one problem to another with one tagged mixture that works uniformly over the decision class. The pairwise Blackwell, Lehmann, and Le Cam results will then translate this uniform value statement into their respective statistical languages.

\begin{lemma}[Source-set lifting principle]
\label{lem:lifting}
Let $P,Q$ be source sets, let $\mathcal C\subseteq\sD$ be closed under finite direct sums, and let $c\geq0$. Then the following statements are equivalent:
\begin{enumerate}[label=(\roman*)]
\item $U(P,D)\geq U(Q,D)-c$ for every $D\in\mathcal C$.
\item For every $q_\eta\in\coapp(Q)$, there is
      $p_\xi\in\coapp(P)$ such that
      $U(p_\xi,D)\geq U(q_\eta,D)-c$ for every $D\in\mathcal C$.
\end{enumerate}
\end{lemma}

\begin{proof}
Assume~(i) and fix $\eta\in\Delta(Q)$. On the compact convex space $K=\Delta(P)$, define
\[
 f_D(\xi)=U(q_\eta,D)-\int_P U(r,D)\xi(dr).
\]
This function is continuous and affine. Direct sums imply that finite convex combinations of these functions are again members of the family. Furthermore,
\[
 \min_{\xi\in\Delta(P)}f_D(\xi)
 =U(q_\eta,D)-U(P,D)
 \leq U(Q,D)-U(P,D)\leq c.
\]
Lemma~\ref{lem:minimax} yields a single $\xi$ for which $f_D(\xi)\leq c$ for every $D\in\mathcal C$. Lemma~\ref{lem:append} gives~(ii).

Conversely, apply~(ii) to each $q\in Q$, viewed as a degenerate lottery. For every $D$,
\[
       U(q,D)-c\leq U(p_\xi,D)\leq U(P,D).
\]
Maximizing over $q$ proves~(i).
\end{proof}

\begin{proof}[Proof of Theorem~\ref{thm:pair}]
Apply Lemma~\ref{lem:lifting} with $P=\{p_1,p_2\}$, $Q=\{q\}$, $\mathcal C=\sD$, and $c=0$. It produces a single $t\in[0,1]$ satisfying $U(p_t,D)\geq U(q,D)$ for all $D$. Lemma~\ref{lem:blackwell} gives $p_t\trianglerighteq q$. The reverse implication is the value calculation displayed immediately
after the theorem statement. The finite-action assertion follows from Lemma~\ref{lem:decisions}(i).
\end{proof}

\begin{example}[Complementarity without individual dominance]
\label{ex:complementarity}
Let $\Omega=\{\omega_1,\omega_2,\omega_3\}$. Source $p_1$ reveals whether the state is $\omega_1$, and $p_2$ reveals whether it is $\omega_3$. Set $q=\tfrac12p_1\oplus\tfrac12p_2$. The pair is complementary relative to $q$ by Theorem~\ref{thm:pair}. Under the uniform prior, consider the binary decision problem that pays one for correctly reporting whether the state is $\omega_1$ and zero otherwise. Values under $(p_1,p_2,q)$ are $(1,2/3,5/6)$. For the analogous problem about $\omega_3$, the values are $(2/3,1,5/6)$. Thus neither $p_1$ nor $p_2$ Blackwell dominates $q$.
\end{example}

\subsection{Comparing Compact Source Sets}

\begin{theorem}[Blackwell comparison of source sets]
\label{thm:sets}
Let $P,Q$ be source sets. The following statements are equivalent:
\begin{enumerate}[label=(\roman*)]
\item For every $D\in\sD$, $U(P,D)\geq U(Q,D)$.
\item For every nonempty compact metric action space $A$,
      \[
       \clco\bigl(\Lambda_P(A)\bigr)
       \supseteq
       \clco\bigl(\Lambda_Q(A)\bigr).
      \]
\item For every $q_\eta\in\coapp(Q)$, there exists
      $p_\xi\in\coapp(P)$ such that $p_\xi\trianglerighteq q_\eta$.
\end{enumerate}
In~(i) and~(ii), it is equivalent to test only finite action
spaces. For finite $A$, the closures in~(ii) can be omitted.
\end{theorem}

\begin{proof}
Lemma~\ref{lem:lifting}, with $\mathcal C=\sD$ and $c=0$, and Lemma~\ref{lem:blackwell} prove the equivalence of~(i) and~(iii). If~(iii) holds, every action distribution feasible under a tagged mixture of $Q$ is feasible under a tagged mixture of $P$. Lemma~\ref{lem:geometry} therefore gives~(ii).

To prove~(ii)$\Rightarrow$(i), expected utility in a fixed $D=(A,u,\pi)$ is the continuous linear functional
\[
       \lambda\longmapsto
       \sum_\omega\pi(\omega)
          \int_A u(\omega,a)\lambda(da\mid\omega).
\]
Its supremum over $\Lambda_P(A)$ is $U(P,D)$. Taking a closed convex hull does not change this supremum, so condition~(ii) implies~(i).

Finite-action value comparisons suffice by Lemma~\ref{lem:decisions}(i). The same linear-functional argument shows that finite-action instances of~(ii) imply all finite-action value comparisons, hence~(i) and the full conclusion. Finally, for finite $A$, $\Lambda_P(A)$ is a compact subset of a finite-dimensional space by Lemma~\ref{lem:geometry}; its convex hull is compact and therefore already closed. The same applies to $Q$.
\end{proof}

The theorem permits uncountably many available sources and arbitrary standard Borel signals, including Gaussian signals. Compactness is imposed on their posterior distributions, rather than on the physical signal spaces. The lotteries in the conclusion may have infinite support. 

\section{Lehmann Comparisons}
\label{sec:lehmann}

\subsection{Choosing a Source for a Monotone Problem}

Suppose the decision maker knows that her decision problem will be monotone. Which of two source sets would she prefer to have available? As in Section~\ref{sec:blackwell}, she learns the problem before choosing one source. The better source may therefore depend on the problem. We begin with this choice criterion and ask whether it has a characterization by comparisons of tagged mixtures.

Throughout this section, $\omega_1<\cdots<\omega_m$. Let $\sM_0\subseteq\sD$ be the class of \emph{monotone decision problems} $D=(A,u,\pi)$ with $A$ a nonempty compact subset of the order-preserving copy of $\IR$ fixed above, which we identify with $\IR$, continuous statewise utilities, and the following property. For every $a'>a$, the differences $d_i=u(\omega_i,a')-u(\omega_i,a)$ satisfy \emph{single crossing from below}: whenever $i<j$, $d_i\geq0$ implies $d_j\geq0$, and $d_i>0$ implies $d_j>0$. Priors $\pi$ are unrestricted. Single crossing restricts the sign of an incremental utility, not its numerical magnitude. It does not require that utility be single-peaked in the action. Familiar examples include choosing treatment intensity as illness severity rises, investment as project quality improves, or inventory as demand conditions strengthen. In each case, a higher action becomes attractive, and remains attractive, as the state increases. This scalar single-crossing formulation is used in the information-comparison literature, including \citet{Persico2000}, \citet{AtheyLevin2018}, and \citet{QuahStrulovici2009}. It is distinct from the original correct-action and loss formulation in \citet[Section~4]{Lehmann1988}; Appendix~\ref{app:original-lehmann} treats that formulation separately. \citet{Kim2023} considers a different extension that permits multidimensional actions.

The economically relevant comparison is
\begin{equation*}
       U(P,D)\geq U(Q,D)\quad\text{for every }D\in\sM_0.\tag{M0}\label{eq:primitive-monotone}
\end{equation*}
For finite source sets, this says that $\max_{p\in P}U(p,D)\geq\max_{q\in Q}U(q,D)$ in every monotone problem. For $P=\{p_1,p_2\}$ and $Q=\{q\}$, it is complementarity relative to $q$ on this restricted decision class. The natural analogue of Theorem~\ref{thm:sets} would replace Blackwell dominance between tagged mixtures by Lehmann dominance. Two different obstacles arise: tagged mixtures may cease to be monotone, and even when they remain monotone, comparison~\eqref{eq:primitive-monotone} need not yield a single lottery that works for all primitive problems.

\subsection{When Tagged Mixtures Preserve MLRP}

For an experiment with ordered real signals and positive conditional densities $f(\cdot\mid\omega_i)$ relative to a common measure, the \emph{monotone likelihood ratio property} (MLRP) means
\[
 f(x\mid\omega_i)f(x'\mid\omega_j)\geq f(x\mid\omega_j)f(x'\mid\omega_i)\quad\text{if }x<x'\text{ and }i<j.
\]
Thus a higher signal raises the likelihood of a higher state relative to a lower one. For finite signals, densities are conditional signal probabilities and the displayed inequality compares columns of the experiment's probability matrix.

For beliefs $\gamma,\gamma'$ with strictly positive coordinates, write $\gamma\leq_{\mathrm{MLR}}\gamma'$ if
\[
 \gamma(\omega_i)\gamma'(\omega_j)\geq\gamma(\omega_j)\gamma'(\omega_i)\quad\text{for every }i<j.
\]
Equivalently, $\gamma'(\omega_i)/\gamma(\omega_i)$ is nondecreasing in $i$. A set of beliefs is an \emph{MLR chain} if every two of its members can be compared in this order. A source satisfying MLRP has an ordered posterior support. The fact that sources individually satisfy MLRP, however, does not order posteriors arising from different sources.

\begin{example}[Tagged mixing can destroy MLRP]
\label{ex:mlrfailure}
Consider two binary-signal experiments with three ordered states:
\[
 p_1=
 \begin{pmatrix}
  0.8&0.2\\
  0.5&0.5\\
  0.2&0.8
 \end{pmatrix}.
\quad\quad
 p_2=
 \begin{pmatrix}
  0.9&0.1\\
  0.4&0.6\\
  0.1&0.9
 \end{pmatrix}.
\]
Rows correspond to states and columns to low and high signals. Both experiments satisfy MLRP. Their high-signal columns $v=(0.2,0.5,0.8)^\top$ and $w=(0.1,0.6,0.9)^\top$ cannot be ordered by likelihood ratios: $w/v=(0.5,1.2,1.125)$ is neither nondecreasing nor nonincreasing. Positive lottery weights do not change this failure. Consequently, for every $t\in(0,1)$, no ordering of the four tagged signals of $t p_1\oplus(1-t)p_2$ satisfies MLRP. Its posterior support contains two incomparable beliefs. Since Blackwell-equivalent sources have the same posterior law (Lemma~\ref{lem:canonical}), it has no Blackwell-equivalent representation satisfying MLRP either.
\end{example}

The example identifies a domain problem for the proposed lifting approach. Restricting each source separately to satisfy MLRP does not make the Lehmann order applicable throughout its tagged hull. We therefore work with source sets whose posterior supports share a common ordering.

\begin{definition}[Admissible source sets]
\label{def:admissible}
A source set $P\subseteq\sE$ is \emph{admissible for Lehmann comparison} if it is nonempty and compact and there is a compact MLR chain $C_P\subseteq\operatorname{int}X$ such that
\[
                  \supp\tau_p\subseteq C_P\quad\text{for every }p\in P.
\]
Here $\operatorname{int}X$ is the relative interior of the belief simplex, consisting of beliefs with positive probability on every state. The support $\supp\tau$ consists of the beliefs whose every open neighborhood has positive $\tau$-probability.
\end{definition}

When $P$ and $Q$ are admissible, each has its own chain $C_P$ or $C_Q$; their union need not be a chain. By Lemma~\ref{lem:append},
\[
       \tau_{p_\xi}=\int_P\tau_p\,\xi(dp).
\]
Moreover, $\supp\tau_{p_\xi}\subseteq C_P$ for every $\xi\in\Delta(P)$. Indeed, every integrand gives probability one to the closed set $C_P$, and so does its average. Thus the tagged hull itself is admissible. The rank construction below shows that every member of this hull has a representation satisfying MLRP.

For a finite collection whose individual posterior supports are compact subsets of $\operatorname{int}X$, the common-chain requirement is exactly the ordering condition needed to make \emph{every} tagged mixture satisfy MLRP. A lottery assigning positive weight to every member has posterior support equal to the union of those supports. If that union contains incomparable beliefs, it cannot have a representation satisfying MLRP; if the union is a chain, it is a suitable compact $C_P$. Example~\ref{ex:mlrfailure} exhibits the former case. For possibly infinite collections, Definition~\ref{def:admissible} also imposes useful compactness and uniform interior regularity: all coordinates on $C_P$ have a common positive lower bound. These topological and positivity conditions are sufficient regularity conditions. We do not assert that they are necessary for every possible analysis of \eqref{eq:primitive-monotone}.

\subsection{The Pairwise Lehmann Order}
We define the ordered comparison using a representation that also treats finite and atomic signals. For any source whose posterior support lies in a compact MLR chain $C\subseteq\operatorname{int}X$, let
\[
       h(\gamma)=\sum_{i=1}^m i\,\gamma(\omega_i).
\]
This function is strictly increasing along distinct MLR-ordered beliefs. To see this, if $r_i=\gamma'(\omega_i)/\gamma(\omega_i)$ is nondecreasing and not constant, then $h(\gamma')-h(\gamma)$ is the strictly positive covariance between $i$ and $r_i$ under $\gamma$. Hence $h$ continuously and one-to-one orders the compact chain; its inverse on $h(C)$ is continuous.

Let $\gamma_p(z)$, $0<z<1$, be the ordered quantile of $\tau_p$: take the usual quantile of the real random variable $h(\gamma)$ under $\tau_p$ and apply $h^{-1}$. For a real distribution with cumulative distribution function $H$, its quantile at $z\in(0,1)$ is $\inf\{x\mid H(x)\geq z\}$. Thus $\gamma_p$ is measurable, nondecreasing in the MLR order, and has law $\tau_p$ when $z$ is uniformly distributed on $(0,1)$. Define a source on $[0,1]$ by the Lebesgue densities
\[
 f_p(z\mid\omega_i)=\frac{\gamma_p(z)(\omega_i)}{\bar\pi(\omega_i)}.
\]
Values at the endpoints are immaterial. This source has unconditional uniform signals under $\bar\pi$, posterior $\gamma_p(z)$, and therefore is Blackwell equivalent to $p$. We call it the \emph{rank representation}. Atoms in the original posterior law become intervals of ranks with the same posterior. The additional randomization on such an interval carries no information about the state.

The rank densities are bounded above and bounded away from zero, and the resulting source satisfies MLRP. Consequently
\[
 F_p(x\mid\omega_i)=\int_0^x f_p(z\mid\omega_i)\,dz
\]
is continuous and strictly increasing from zero to one on $[0,1]$. Its ordinary inverse is well defined. No choice of an inverse at an atom is needed.

For two sources admitting these representations, define $p\trianglerighteq_L q$ if
\begin{equation}
 i\longmapsto F_p^{-1}\bigl(F_q(y\mid\omega_i)\mid\omega_i\bigr)
 \quad\text{is nondecreasing for every }y\in(0,1).
 \label{eq:lehmann}
\end{equation}
This is the quantile formulation of the Lehmann order, with the more informative source on the left. It is unchanged by increasing transformations of the ranked signals. An equivalent form is
\begin{equation}
 \psi_i(x):=F_q^{-1}\bigl(F_p(x\mid\omega_i)\mid\omega_i\bigr)\quad\text{is nonincreasing in }i \quad\text{for every }x\in(0,1).
 \label{eq:quantiletransport}
\end{equation}
Indeed, set $y=\psi_i(x)$ in \eqref{eq:lehmann} and apply the strictly increasing distribution functions and their inverses; the resulting inequality is $\psi_j(x)\leq\psi_i(x)$ for $j>i$. The reverse argument gives the converse. Under MLRP, \citet[Theorem~1]{Kim2023} shows that this comparison is equivalent to $q$ being a monotone quasi-garbling of $p$: relative to $p$, $q$ is generated by state-dependent noise that tends to move reported signals downward in higher states and upward in lower states, thereby weakening the alignment between signals and states.

Related pairwise results for single-crossing and monotone decision problems appear in \citet{Persico2000}, \citet{AtheyLevin2018}, \citet{QuahStrulovici2009}, and \citet{Kim2023}. Lehmann's original theorem cannot be invoked directly here because it concerns the correct-action loss class discussed in Appendix~\ref{app:original-lehmann}, not the single-crossing class $\sM_0$. \citet[Theorem~1]{Persico2000} gives the corresponding equivalence for affiliated signals and single-crossing payoffs, while \citet[Corollary~1]{Kim2023} provides a general necessary-and-sufficient result for generally monotone decision problems. The rank representation above places our finite-state sources, including sources with atomic signals, within Kim's regular signal framework. The next lemma records the resulting specialization in our notation.

\begin{lemma}[Pairwise Lehmann Comparison]
\label{lem:lehmann}
Let the posterior supports of $p$ and $q$ each lie in a compact MLR chain in $\operatorname{int}X$. Then
\[
       p\trianglerighteq_L q\quad\Longleftrightarrow\quad U(p,D)\geq U(q,D)\text{ for every }D\in\sM_0.
\]
\end{lemma}

The proof is in Appendix~\ref{app:proof-lehmann}.

\subsection{Why Primitive Comparison Does Not Lift}

On admissible source sets, all tagged mixtures belong to the domain of the pairwise order. The candidate condition that every tagged mixture from $Q$ be Lehmann dominated by some tagged mixture from $P$ is sufficient for \eqref{eq:primitive-monotone}. Indeed, apply it to each $q\in Q$ and use Lemma~\ref{lem:lehmann} and the affinity of tagged mixtures to obtain
\[
       U(q,D)\leq U(p_\xi,D)\leq U(P,D)\quad\text{for every }D\in\sM_0.
\]
The question is whether the converse holds. The next example answers no, even with two sources, a single benchmark, and finite signals.

\begin{example}[A Four-State Source Set]
\label{ex:four-state}
Let $\Omega=\{\omega_1<\omega_2<\omega_3<\omega_4\}$ and
$P=\{p_1,p_2\}$, $Q=\{q\}$, where
\[
 p_1=\frac1{10}
 \begin{pmatrix}8&2\\6&4\\2&8\\1&9\end{pmatrix},
\quad\quad
 p_2=\frac1{10}
 \begin{pmatrix}9&1\\8&2\\4&6\\2&8\end{pmatrix},
\quad\quad
 q=\frac1{10}
 \begin{pmatrix}8&2\\6&4\\4&6\\2&8\end{pmatrix}.
\]
Columns are low and high signals. All entries are positive and all three experiments satisfy strict MLRP. The posterior supports of $p_1,p_2$ under the uniform reference prior form the chain
\[
 \frac{(8,6,2,1)}{17}
 \leq_{\mathrm{MLR}}\frac{(9,8,4,2)}{23}
 \leq_{\mathrm{MLR}}\frac{(2,4,8,9)}{23}
 \leq_{\mathrm{MLR}}\frac{(1,2,6,8)}{17}.
\]
For example, the ratios of successive unnormalized vectors are $(9/8,4/3,2,2)$, $(2/9,1/2,2,9/2)$, and $(1/2,1/2,3/4,8/9)$; each is nondecreasing. Normalization multiplies each ratio vector by a positive constant. Take $C_P$ to be these four beliefs and
\[
       C_Q=\left\{\frac{(4,3,2,1)}{10},\frac{(1,2,3,4)}{10}\right\}.
\]
Both are finite MLR chains in $\operatorname{int}X$. Hence $P,Q$ are admissible and every $p_t=t p_1\oplus(1-t)p_2$ has a representation satisfying MLRP.
\end{example}
The point of this example is summarized in the following claim:
\begin{claim}
\label{clm:four-state}
Let $p_1,p_2,q$ be constructed as in the previous example.
\begin{enumerate}[label=(\alph*)]
\item $\max\{U(p_1,D),U(p_2,D)\}\geq U(q,D)$ for every
      $D\in\sM_0$.
\item No $t\in[0,1]$ satisfies $U(p_t,D)\geq U(q,D)$ for every
      $D\in\sM_0$.
\end{enumerate}
\end{claim}
The proof is in Appendix~\ref{app:proof-four-state}.

The conflict in part~(b) is already visible in two binary-action problems $D_1,D_2$ with the uniform prior and $u_k(\omega_i,0)=0$. Because part~(a) holds on the larger class with unrestricted priors, the same example remains a counterexample if every problem in $\sM_0$ is required to share this uniform full-support prior. In the first problem,
\[
       (u_1(\omega_i,1))_{i=1}^4=(-6,6,1,1).
\]
In the second,
\[
       (u_2(\omega_i,1))_{i=1}^4=(-1,-1,-6,6).
\]
Both satisfy single crossing. Their values are
\[
 \begin{array}{c|ccc}
       &U(p_1,D_k)&U(p_2,D_k)&U(q,D_k)\\ \hline
 D_1   &29/40&1/2&13/20\\
 D_2   &0&9/40&3/20
 \end{array}
\]
By the affine value of tagged mixtures, matching $q$ in $D_1$ requires $t\geq2/3$, whereas matching it in $D_2$ requires $t\leq1/3$. The two problems therefore demand incompatible source frequencies.

The example separates the two roles played by the sources.  Source $p_1$ is the useful specialist in $D_1$, while $p_2$ is the useful specialist in $D_2$.  The primitive source set comparison permits the decision maker to switch between them after seeing which decision problem she faces.  A tagged mixture fixes the relative frequency with which the specialists are available, and the two decision problems demand incompatible frequencies.  Because all posterior supports lie on the required MLR chains, this conflict is not caused by a failure of monotone signal ordering.

By Lemma~\ref{lem:lehmann}, no tagged mixture constructed in Example~\ref{ex:four-state} dominates $q$ in the Lehmann order. Nor can one Blackwell dominate $q$, since Blackwell dominance would imply the two value inequalities in the table above. Thus replacing Lehmann dominance by Blackwell dominance does not characterize \eqref{eq:primitive-monotone}. More generally, redefining the pairwise order cannot repair the proposed lifting equivalence on this domain.

\begin{proposition}[No universal pairwise lifting of primitive comparison]
\label{prop:no-lifting}
There is no binary relation $R$ between sources with compact MLR posterior supports in the relative interior of the relevant belief simplex such that, for every finite ordered state space and every pair of finite admissible source sets $P,Q$, comparison~\eqref{eq:primitive-monotone} is equivalent to
\begin{equation}
       \text{for every }q_\eta\in\coapp(Q),\text{ there is }
       p_\xi\in\coapp(P)\text{ with }p_\xi\,R\,q_\eta.
       \label{eq:proposed-lifting}
\end{equation}
\end{proposition}

\begin{proof}
Apply the proposed equivalence to singleton source sets $P=\{p\}$ and $Q=\{q\}$. Their tagged hulls are the same singleton source sets, so it forces
\[
       p\,R\,q\quad\Longleftrightarrow\quad U(p,D)\geq U(q,D)\ \text{for every }D\in\sM_0.
\]
In particular, $R$ must coincide with the pairwise Lehmann order on each such domain. Now use the admissible source sets in Example~\ref{ex:four-state}. Claim~\ref{clm:four-state}(a) and the proposed equivalence would give some $p_t\,R\,q$. Since $\{p_t\}$ is itself admissible, the singleton implication would give dominance in every primitive problem, contradicting Claim~\ref{clm:four-state}(b).
\end{proof}

The proposition rules out any fixed order of the proposed form, including one positioned between Blackwell and Lehmann dominance. It does not rule out characterizations that use the joint structure of source sets, or results on smaller domains. What fails is the Blackwell-style exchange between a source chosen separately for each primitive problem and one source lottery that succeeds in every primitive problem. Restricting to admissible source sets repairs the MLR domain but does not repair this quantifier exchange.

The obstruction is visible in decision values themselves. In Example~\ref{ex:four-state}, let $g(D)=(U(p_1,D)-U(q,D),U(p_2,D)-U(q,D))$. Then $g(D_1)=(3,-6)/40$ and $g(D_2)=(-6,3)/40$, while their average is $(-3,-3)/80$. No primitive problem has both coordinates of $g(D)$ negative, by Claim~\ref{clm:four-state}(a). Its attainable value profiles are therefore nonconvex. This is stronger than observing that a literal product of scalar action sets need not satisfy single crossing: even a different primitive problem cannot reproduce these averaged values at the three displayed sources.

This pinpoints the difference from \citet[Theorem~2(iv)]{ChengBorgers2026}. Their result assumes a convex set of situations and utilities linear in the situation. Here the alternatives $\{p_1,p_2,q\}$ are finite and evaluation of a value profile is a linear coordinate projection; convexity of attainable profiles fails. Accordingly, $q$ is neither weakly dominated by a lottery over the other sources nor redundant, yet it is never uniquely optimal on $\sM_0$. Their Blackwell application obtains convexity by constructing contingent action products. Their Proposition~3 uses finite lattice action spaces and jointly supermodular utilities, a different monotone class that also permits this construction. They explicitly distinguish its informativeness order from Lehmann's. The scalar restriction here prevents that argument from applying.

\citet{Kim2023} studies pairwise comparison on another monotone class that permits multidimensional actions and characterizes the comparison through monotone quasi-garbling. The bundles below also allow action profiles with several coordinates, but they are defined by additive separability across primitive problems. It is this direct-sum closure, rather than multidimensionality by itself, that restores the source-set lifting argument.

\subsection{Bundles of Monotone Problems and the Lifting Theorem}

The preceding failure suggests an explicit extension of the economic question. Suppose one source must inform a finite bundle of monotone decision problems, rather than one primitive decision problem. Let $\sM$ be the class of finite weighted direct sums of members of $\sM_0$, constructed as in Lemma~\ref{lem:decisions}(ii), including single-member sums. Thus, for $D=\bigoplus_{k=1}^n\theta_kD_k$, where $D_k\in\sM_0$ and $\theta\in\Delta(\{1,\ldots,n\})$,
\[
       U(p,D)=\sum_k\theta_kU(p,D_k).
\]
For a source set, the corresponding value is
\[
       U(P,D)=\max_{p\in P}\sum_k\theta_kU(p,D_k).
\]
The decision maker knows the bundle when selecting her source, and can condition each decision problem's action on its signal. She cannot select a different source for each decision problem. Decision problems with a common prior have the direct interpretation of several choices informed by the same source; the prior-weighted construction also allows distinct component priors.

For individual sources, comparing all of $\sM$ is equivalent to comparing all of $\sM_0$: pairwise inequalities can be added, and primitive problems are one-problem bundles. For source sets, this enlargement is a substantive strengthening of \eqref{eq:primitive-monotone}. Indeed, for the equal bundle of $D_1,D_2$ in Example~\ref{ex:four-state},
\[
       U(P,\tfrac12D_1\oplus\tfrac12D_2)=\frac{29}{80}<\frac25=U(q,\tfrac12D_1\oplus\tfrac12D_2).
\]
The source that performs well in one decision problem need not be the one that performs well in the other. Bundling restores convexity of the value profiles by imposing comparisons on their finite convex combinations. It is the finite direct-sum closure of the primitive class, and a direct sum of bundles is again a bundle.

\begin{theorem}[Lehmann lifting for admissible source sets]
\label{thm:lehmann}
Let $P,Q\subseteq\sE$ be admissible source sets. The following statements are equivalent:
\begin{enumerate}[label=(\roman*)]
\item $U(P,D)\geq U(Q,D)$ for every bundle $D\in\sM$ of monotone
      decision problems.
\item For every $q_\eta\in\coapp(Q)$, there exists
      $p_\xi\in\coapp(P)$ such that
      $p_\xi\trianglerighteq_L q_\eta$.
\end{enumerate}
The source sets may be infinite. All tagged mixtures in (ii) admit the rank representations used to define the Lehmann order.
\end{theorem}

\begin{proof}
By admissibility and Lemma~\ref{lem:append}, each tagged mixture from $P$ has posterior support in $C_P$, and each tagged mixture from $Q$ has posterior support in $C_Q$. Hence Lemma~\ref{lem:lehmann} applies to every pair of tagged mixtures in (ii).

Suppose (i) holds. The class $\sM$ is closed under finite direct sums, so Lemma~\ref{lem:lifting}, with $\mathcal C=\sM$ and $c=0$, gives for each $q_\eta$ a single $p_\xi$ satisfying $U(p_\xi,D)\geq U(q_\eta,D)$ for all $D\in\sM$. In particular, this holds for all $D\in\sM_0$. Lemma~\ref{lem:lehmann} implies $p_\xi\trianglerighteq_Lq_\eta$.

Conversely, suppose (ii) holds and fix $q\in Q$, regarded as a degenerate tagged mixture. The corresponding $p_\xi$ dominates $q$ on every primitive problem by Lemma~\ref{lem:lehmann}, and therefore on every bundle by addition. For every $D\in\sM$,
\[
       U(q,D)\leq U(p_\xi,D)=\int_P U(p,D)\xi(dp)\leq U(P,D).
\]
Maximizing over $q\in Q$ proves (i).
\end{proof}

Taking $P=\{p_1,p_2\}$ and $Q=\{q\}$ gives complementarity on monotone decision problem bundles. Definition~\ref{def:admissible} and the bundle criterion address the two obstacles separately: common posterior chains keep tagged mixtures monotone; bundles make the value comparisons closed under finite combinations. Theorem~\ref{thm:lehmann} answers this strengthened choice question. Example~\ref{ex:four-state}, Claim~\ref{clm:four-state}, and Proposition~\ref{prop:no-lifting} explain why its pairwise lifting conclusion cannot instead characterize the original primitive criterion~\eqref{eq:primitive-monotone} throughout the admissible domain.

\section{Le Cam Comparisons}
\label{sec:lecam}

Exact Blackwell dominance is deliberately demanding: any decision problem can witness a failure. For applications such as replacing a data provider or approximating a source set of diagnostic tests, it is natural to ask how much value can be lost in the worst case. A quantitative answer requires a common payoff scale. We use Le Cam's randomization criterion and deficiency; \citet{Torgersen1991} gives a general treatment of this theory.

Let $\sD_1\subseteq\sD$ consist of the problems with $0\leq u(\omega,a)\leq1$ for all $(\omega,a)$. Priors are still unrestricted. This class is closed under finite direct sums by Lemma~\ref{lem:decisions}. Without a utility normalization, any strictly positive value gap could be made arbitrarily large by rescaling.

\subsection{Total Variation and Approximate Garblings}

Economically, deficiency asks how well a decision maker who observes $p$ can imitate one who observes $q$, allowing her to randomize after seeing her signal. It is the smallest worst-state error in that imitation. Zero deficiency recovers an exact Blackwell garbling; positive deficiency measures the largest loss that remains after the best randomized simulation. The proposition below makes this interpretation exact for bounded payoffs.

For two probability measures on the same measurable space, their \emph{total variation distance} is
\begin{equation}
 \TV(\mu,\nu)=\sup_{B}|\mu(B)-\nu(B)|,
 \label{eq:tv}
\end{equation}
where the supremum ranges over measurable sets. With this convention, $0\leq\TV\leq1$; on a finite space it equals $\tfrac12\sum_s|\mu(s)-\nu(s)|$. For any measurable $v$ with $0\leq v\leq1$,
\begin{equation}
 \left|\int v\,d\mu-\int v\,d\nu\right|
       \leq\TV(\mu,\nu).
 \label{eq:tv-utility}
\end{equation}
Indeed, write $v(s)=\int_0^1\ind_{\{v(s)>t\}}\,dt$ and use \eqref{eq:tv}. In particular, total variation contracts under kernels:
\[
       \TV(K\circ\mu,K\circ\nu)\leq\TV(\mu,\nu),
\]
by applying \eqref{eq:tv-utility} to $v(s)=K(B\mid s)$.

The directed \emph{Le Cam deficiency} of $p$ relative to
$q$ is
\begin{equation}
 \delta(p,q)=\inf_{K:S\to\Delta(T)}
       \max_{\omega\in\Omega}
       \TV\bigl((K\circ p)(\cdot\mid\omega),
                      q(\cdot\mid\omega)\bigr).
 \label{eq:deficiency}
\end{equation}
Thus $\delta(p,q)$ is the smallest uniform state-by-state error in simulating $q$ from $p$. Write
$p\trianglerighteq_\varepsilon q$ if $\delta(p,q)\leq\varepsilon$. The orientation is the same as the Blackwell order: a small deficiency means that $p$ can reproduce $q$ accurately. Composition adds error tolerances, so at positive $\varepsilon$ this is an approximate comparison, rather than a transitive order with a fixed tolerance.

By composing kernels and using total-variation contraction, replacing either source by a Blackwell-equivalent one does not change $\delta(p,q)$. We may therefore use the canonical signal space $X$ in its analysis.

\begin{proposition}[Pairwise Randomization Formula]
\label{prop:deficiency}
For any sources $p,q$, the infimum in \eqref{eq:deficiency} is attained, and
\begin{equation}
       \delta(p,q)=\sup_{D\in\sD_1}\bigl[U(q,D)-U(p,D)\bigr].
 \label{eq:randomization}
\end{equation}
The supremum is unchanged if only finite-action problems in $\sD_1$ are used.
\end{proposition}

The proof is in Appendix~\ref{app:proof-deficiency}.

The use of all priors in $\sD_1$ corresponds to the maximum over states in \eqref{eq:deficiency}. Fixing one prior would give a different, prior-weighted approximation criterion. The total variation convention in \eqref{eq:tv} and the utility interval $[0,1]$ are responsible for the constant one in \eqref{eq:randomization}.

\subsection{The Exact Shortfall of a Source Set}

For source sets $P,Q$, define the directed deficiency between their tagged hulls by
\begin{equation}
 \delta^\oplus(P,Q)=\sup_{q_\eta\in\coapp(Q)}\min_{p_\xi\in\coapp(P)}\delta(p_\xi,q_\eta).
 \label{eq:setdeficiency}
\end{equation}
The inner minimum exists. Indeed, Proposition~\ref{prop:deficiency} expresses $p\mapsto\delta(p,q)$ as a supremum of continuous functions of $p$, so it is lower semicontinuous on the compact set $\coapp(P)$.

\begin{theorem}[Exact deficiency formula for source sets]
\label{thm:epsilon}
Let $P,Q$ be source sets. Then
\begin{equation}
 \delta^\oplus(P,Q)=\sup_{D\in\sD_1}\bigl[U(Q,D)-U(P,D)\bigr].
 \label{eq:setrandomization}
\end{equation}
In particular, for every $\varepsilon\geq0$, the following are equivalent:
\begin{enumerate}[label=(\roman*)]
\item $U(P,D)\geq U(Q,D)-\varepsilon$
      for every $D\in\sD_1$.
\item For every $q_\eta\in\coapp(Q)$, there exist
      $p_\xi\in\coapp(P)$ and a Markov kernel $K$
      from its signal space to that of $q_\eta$ such that
      \[
       \max_{\omega\in\Omega}
       \TV\bigl((K\circ p_\xi)(\cdot\mid\omega),
                         q_\eta(\cdot\mid\omega)\bigr)
             \leq\varepsilon.
      \]
\end{enumerate}
Finite-action problems suffice both in
\eqref{eq:setrandomization} and in~(i).
\end{theorem}

\begin{proof}
Let $v=\sup_{D\in\sD_1}[U(Q,D)-U(P,D)]$. Because values lie in $[0,1]$ and the zero-utility problem is allowed, $v\in[0,1]$. The class $\sD_1$ is closed under finite direct sums. Lemma~\ref{lem:lifting}, applied with $c=v$, therefore gives, for every $q_\eta$, a tagged mixture $p_\xi$ satisfying
\[
       U(q_\eta,D)-U(p_\xi,D)\leq v\quad\text{for every }D\in\sD_1.
\]
Proposition~\ref{prop:deficiency} implies $\delta(p_\xi,q_\eta)\leq v$. Hence $\delta^\oplus(P,Q)\leq v$.

Conversely, put $d=\delta^\oplus(P,Q)$. For every $q\in Q$, choose a minimizer in \eqref{eq:setdeficiency}, treating $q$ as a degenerate tagged mixture. It satisfies $\delta(p_\xi,q)\leq d$. For every $D\in\sD_1$, Proposition~\ref{prop:deficiency} and affinity then give
\[
       U(q,D)\leq U(p_\xi,D)+d\leq U(P,D)+d.
\]
Maximizing over $q$, and then taking the supremum over $D$, yields $v\leq d$. This proves \eqref{eq:setrandomization}.

Condition~(i) is equivalent to $v\leq\varepsilon$. By \eqref{eq:setrandomization}, this is equivalent to $\delta^\oplus(P,Q)\leq\varepsilon$. The inner mixture minimum in \eqref{eq:setdeficiency} and the kernel minimum in Proposition~\ref{prop:deficiency} are both attained, which gives exactly~(ii). Finite-action sufficiency follows from Lemma~\ref{lem:decisions}(i).
\end{proof}

\begin{corollary}[Approximate two-source complementarity]
\label{cor:epsilonpair}
For any sources $p_1,p_2,q$,
\[
 \begin{split}
 &\min_{t\in[0,1]}
     \delta\bigl(t p_1\oplus(1-t)p_2,q\bigr)
 =\sup_{D\in\sD_1}
     \left[U(q,D)-\max\{U(p_1,D),U(p_2,D)\}\right].
 \end{split}
\]
Thus a pair's maximal normalized shortfall relative to $q$ is precisely the smallest simulation error of a disclosed lottery over the pair.
\end{corollary}

\begin{proof}
Apply Theorem~\ref{thm:epsilon} to $P=\{p_1,p_2\}$ and $Q=\{q\}$.
\end{proof}

At $\varepsilon=0$, an attained zero-deficiency kernel is an exact Blackwell garbling. Moreover, every nonconstant bounded utility function can be rescaled to $[0,1]$ by a positive affine transformation, which preserves exact value comparisons. Theorem~\ref{thm:epsilon} and Corollary~\ref{cor:epsilonpair} therefore recover the exact results of Section~\ref{sec:blackwell}. No accumulation of the tolerance occurs when forming decision problem bundles: their weights sum to one.

\section{Conclusion}
\label{sec:conclusion}

A collection of specialists can be uniformly as valuable as a generalist even when no specialist is uniformly better.  Under unrestricted Bayesian decision making, this portfolio value has an exact statistical certificate: a disclosed lottery over the specialists Blackwell dominates the benchmark. For two sources the same lottery works across every decision problem; for compact source sets the comparison is between their tagged hulls.  The deficiency formula puts a cardinal version of the same statement in payoff units.

The ordered analysis shows where this uniformity comes from.  Unrestricted decision problems can bundle any finite collection of decision problems, so a source set that wins problem by problem must also survive every weighted combination of those problems.  Primitive scalar-action monotone problems do not have that closure. Even after imposing a common posterior order so that tagged mixtures remain MLRP, a source set may cover each decision problem separately while no source lottery covers all decision problems at once.  Once one source is required to serve a bundle of monotone decision problems, the Lehmann characterization returns.

Whether a source set can be reduced to one statistically superior randomized source therefore depends on how the information will be used.  An organization that can choose a new specialist for every decision problem has a different information technology from one that must contract with a common provider, maintain a single database, or use one diagnostic platform across several decision problems. The value of an information source set is determined not only by the experiments it contains, but also by the scope and timing of the decision problems those experiments must inform.

\appendix
\section{Omitted Proofs}
\label{app:omitted-proofs}

\subsection{Proof of Lemma~\ref{lem:lehmann}}
\label{app:proof-lehmann}

The rank representations of $p$ and $q$ have continuous, strictly increasing conditional distribution functions and satisfy MLRP. We verify the hypotheses of \citet[Corollary~1]{Kim2023}. For every $D=(A,u,\pi)\in\sM_0$, MLRP and single crossing imply that the largest optimal action can be chosen nondecreasing in the signal for every prior; this is the standard monotone-selection argument (see, e.g., \citealp{QuahStrulovici2009}). Thus both rank representations satisfy Kim's monotone comparative statics condition with respect to every problem in $\sM_0$.

Single crossing also gives Kim's dominated decreasing decision rule condition. Suppose $a_1\geq\cdots\geq a_m$ are assigned to increasing states. Some action in $\{a_1,\ldots,a_m\}$ weakly improves on this assignment in every state. To see this, proceed by induction. Let $a$ work for the first $m-1$ states. Since $a\geq a_m$, either $a$ also works in state $m$, or single crossing implies that $a_m$ improves on $a$ in every earlier state and hence works in all states. Finally, for every state cutoff and every $\kappa\in(0,1)$, the binary problem whose incremental payoff is $-\kappa$ below the cutoff and $1-\kappa$ at and above it belongs to $\sM_0$. Hence $\sM_0$ contains all the simple hypothesis-testing problems used in Kim's necessity argument.

It follows that $\sM_0$ is generally monotone with respect to both rank representations and satisfies the richness condition in \citet[Corollary~1]{Kim2023}. That result gives equivalence between Lehmann dominance and value dominance over all problems in $\sM_0$. Since each source is Blackwell equivalent to its rank representation, the same equivalence holds for $p$ and $q$. This completes the proof.

\subsection{Proof of Claim~\ref{clm:four-state}}
\label{app:proof-four-state}

First take a finite action set and any prior $\pi$. The monotone-selection argument in the preceding proof gives optimal actions $a_L\leq a_H$ for the two signals of $q$. If they coincide, either source can use the same constant plan. Otherwise, let $d_i=u(\omega_i,a_H)-u(\omega_i,a_L)$. Using this plan under $p_1$ rather than $q$ changes expected utility by
\[
 \frac15\pi(\omega_3)d_3+\frac1{10}\pi(\omega_4)d_4.
\]
Using it under $p_2$ rather than $q$ changes expected utility by
\[
 -\frac1{10}\pi(\omega_1)d_1-\frac15\pi(\omega_2)d_2.
\]
If $d_3\geq0$, single crossing gives $d_4\geq0$, so $p_1$ does at least as well as $q$. If $d_3<0$, single crossing gives $d_1,d_2<0$, so $p_2$ does at least as well. Optimizing under the selected source can only increase its value. This proves part~(a) for finite actions. Restricting a compact action set to the finite approximating sets in Lemma~\ref{lem:decisions}(i) preserves single crossing, and the values converge uniformly across sources. Hence part~(a) holds throughout $\sM_0$.

For part~(b), use the two problems $D_1,D_2$ displayed after the claim. For a binary-signal source $e$, their values are calculated from
\[
 U(e,D_k)=\frac14\sum_{s\in\{L,H\}}\max\left\{0,\sum_{i=1}^4e(s\mid\omega_i)u_k(\omega_i,1)\right\}.
\]
Substitution gives the value table in the text. By affinity, matching $q$ in $D_1$ requires $t\geq2/3$, whereas matching it in $D_2$ requires $t\leq1/3$. No $t$ satisfies both requirements. This completes the proof.

\subsection{Proof of Proposition~\ref{prop:deficiency}}
\label{app:proof-deficiency}

First use canonical signals. On a compact metric space, total variation also satisfies
\[
 \TV(\mu,\nu)=\sup_{\substack{v\in C(X)\\0\leq v\leq1}}\int_Xv\,d(\mu-\nu).
\]
To justify the restriction to continuous functions, use regularity of finite Borel measures to approximate a Borel set from within by a compact set and from outside by an open set, then approximate its indicator by a continuous function between zero and one. The absolute value can be omitted because $\mu-\nu$ has total mass zero: replacing $v$ by $1-v$ reverses the sign.

Total variation is therefore lower semicontinuous in the weak topology. Lemma~\ref{lem:geometry} makes $\Lambda_p(X)$ compact, so
\[
 \delta(p,q)=\min_{\lambda\in\Lambda_p(X)}\max_{\omega\in\Omega}\TV\bigl(\lambda(\cdot\mid\omega),q(\cdot\mid\omega)\bigr)
\]
has a minimizer. It is implemented by a kernel. Composing with the equivalence kernels in Lemma~\ref{lem:canonical} gives attainment on the original signal spaces as well.

For one inequality in \eqref{eq:randomization}, let $K$ attain the deficiency. Given an action plan $\alpha$ under $q$, use $\alpha\circ K$ under $p$. Its conditional payoff lies in $[0,1]$, so \eqref{eq:tv-utility} and averaging by $\pi$ imply
\[
 U(q,D)-U(p,D)\leq\delta(p,q)\quad\text{for every }D\in\sD_1.
\]

For the reverse inequality, let $g=(g_\omega)_\omega$ range over tuples of nonnegative continuous functions on $X$ satisfying $\sum_\omega\|g_\omega\|_\infty\leq1$. For $\lambda\in\Lambda_p(X)$, define
\[
 f_g(\lambda)=\sum_\omega\int_Xg_\omega(a)[q(da\mid\omega)-\lambda(da\mid\omega)].
\]
The continuous-function representation of total variation gives
\[
 \sup_g f_g(\lambda)=\max_{\omega\in\Omega}\TV\bigl(q(\cdot\mid\omega),\lambda(\cdot\mid\omega)\bigr).
\]
The upper bound allocates the available test weight across states; the lower bound puts all weight on a state attaining the maximum. The functions $f_g$ are continuous and affine and are closed under finite convex combinations. Lemma~\ref{lem:minimax} therefore gives
\[
 \delta(p,q)=\sup_g\left[\sum_\omega\int_Xg_\omega(a)q(da\mid\omega)-\max_{\lambda\in\Lambda_p(X)}\sum_\omega\int_Xg_\omega(a)\lambda(da\mid\omega)\right].
\]

For each $g$, choose a prior $\pi$ with $\pi(\omega)\geq\|g_\omega\|_\infty$ in every state, distributing any remaining mass arbitrarily. Set $u(\omega,a)=g_\omega(a)/\pi(\omega)$ when $\pi(\omega)>0$, and set it to zero otherwise. This defines a problem in $\sD_1$. In the last display, the first term is the payoff of the identity plan under $q$ and is at most $U(q,D)$; the second is $U(p,D)$. Taking suprema gives the reverse inequality in \eqref{eq:randomization}. Finally, Lemma~\ref{lem:decisions}(i) preserves the $[0,1]$ utility bound and approximates each value gap uniformly, so finite actions suffice. This completes the proof.

\section{Canonical Representation of General Experiments}
\label{app:canonical}

This appendix supplies the construction used in Section~\ref{sec:setup}.  It is included to make the treatment of arbitrary standard Borel signal spaces self-contained.

For finite measures $\nu$ and $\mu$ on the same measurable space, $\nu\ll\mu$ means that $\nu$ is \emph{absolutely continuous} with respect to $\mu$: $\mu(B)=0$ implies $\nu(B)=0$ for every measurable $B$.  Its Radon--Nikodym derivative $d\nu/d\mu$ is a measurable nonnegative function
satisfying
\[
 \nu(B)=\int_B\frac{d\nu}{d\mu}\,d\mu\quad\text{for every measurable }B.
\]
The derivative is unique outside a $\mu$-null set.

Given a source $p$, define its unconditional signal law under the reference prior by
\[
 \mu_p=\sum_\omega\bar\pi(\omega)p(\cdot\mid\omega).
\]
Its posterior map is
\[
 b_p(s)(\omega)=\bar\pi(\omega)\frac{dp(\cdot\mid\omega)}{d\mu_p}(s).
\]
The derivative exists because $p(\cdot\mid\omega)\ll\mu_p$.  The coordinates of $b_p$ sum to one almost everywhere; choose any belief on the common exceptional null set.  The posterior distribution used in the text is
\[
 \tau_p(B)=\mu_p\{s\mid b_p(s)\in B\}.
\]
This definition applies to every $B\in\sB(X)$. It satisfies the Bayes-plausibility condition $\int_X\gamma\,\tau_p(d\gamma)=\bar\pi$.

\begin{proof}[Proof of Lemma~\ref{lem:canonical}]
The deterministic posterior map garbles $p$ into $\widehat p_{\tau_p}$, since
\[
 p(b_p^{-1}(B)\mid\omega)=\frac{1}{\bar\pi(\omega)}\int_{b_p^{-1}(B)}b_p(s)(\omega)\mu_p(ds)=\widehat p_{\tau_p}(B\mid\omega).
\]
Standard Borel spaces admit regular conditional probabilities.  In particular, there is a kernel $L:X\to\Delta(S)$ giving the conditional law of the original signal given its posterior and satisfying
\[
 \int_B L(C\mid\gamma)\tau_p(d\gamma)=\mu_p\{s\in C\mid b_p(s)\in B\}
\]
for measurable $B\subseteq X$ and $C\subseteq S$.  Integrating the coordinate $\gamma(\omega)$ in this identity gives
\[
 \int_X L(C\mid\gamma)\widehat p_{\tau_p}(d\gamma\mid\omega)=\frac{1}{\bar\pi(\omega)}\int_C b_p(s)(\omega)\mu_p(ds)=p(C\mid\omega).
\]
Hence the canonical source also garbles into $p$.

For uniqueness, suppose a kernel garbles $p$ into $q$ and generate $(\omega,s,t)$ using $\bar\pi$, $p$, and that kernel.  With $B=b_p(s)$ and $C=b_q(t)$, conditional expectation gives $C=\E[B\mid t]$.  Therefore
\[
 \int_X\|\gamma\|^2\tau_p(d\gamma)-\int_X\|\gamma\|^2\tau_q(d\gamma)=\E\|B-C\|^2\geq0,
\]
where $\|\cdot\|$ is the Euclidean norm.  If a reverse garbling exists, the two integrals are equal, so $B=C$ almost surely in the first construction and $\tau_p=\tau_q$.  The converse follows from their common canonical representation.
\end{proof}

\section{Lehmann's Original Monotone Decision Problems}
\label{app:original-lehmann}

The primitive class in Section~\ref{sec:lehmann} is defined by single crossing of utility differences. \citet[Section~4]{Lehmann1988} instead starts with a nondecreasing correct-action map $a^*:\Omega\to\IR$ and takes its image $A=a^*(\Omega)$ as the action space. A nonnegative loss $L:\Omega\times A\to\IR_+$ is zero at the correct action and increases weakly as the chosen action moves away from it on either side. Precisely, $L(\omega_i,a^*(\omega_i))=0$ for every state $\omega_i$; if $a\leq b\leq a^*(\omega_i)$, then $L(\omega_i,a)\geq L(\omega_i,b)$; and if $a^*(\omega_i)\leq a\leq b$, then $L(\omega_i,a)\leq L(\omega_i,b)$. Denote by $\sM_{\mathrm{Leh}}$ the Bayesian problems with these actions and losses, $u=-L$, and arbitrary priors. Adding a constant to utility in each state does not change any source comparison. With finite $\Omega$, the restriction $A=a^*(\Omega)$ entails $|A|\leq|\Omega|$; every feasible action is correct in some state.

The loss restriction is a shape condition within each state, with correct actions ordered across states. It is not the single-crossing condition in Section~\ref{sec:lehmann}. The distinction between single crossing and quasiconcavity with increasing peaks is also emphasized by \citet{QuahStrulovici2009}. For example, the utilities
\[
 \begin{array}{c|rrr}
       &a=0&a=1&a=2\\ \hline
 \omega_1&0&-6&-7\\
 \omega_2&0&6&5\\
 \omega_3&0&1&-5\\
 \omega_4&0&1&7
 \end{array}
\]
have correct actions $(0,1,1,2)$. Subtracting utility from its statewise maximum gives losses satisfying the preceding restrictions, but the utility difference between actions two and zero is $(-7,5,-5,7)$, which violates single crossing. For the sources in Example~\ref{ex:four-state} and the uniform prior, the displayed utility normalization gives
\[
              U(q,D)=\frac45>\frac{29}{40} =U(p_1,D)=U(p_2,D).
\]
Thus that example cannot simply be reused for $\sM_{\mathrm{Leh}}$. Conversely, the state-independent utility $u(\omega_i,a)=a^2$ on $A=\{-1,0,1\}$ satisfies single crossing, but its central dip prevents a loss representation of the required shape. A separate counterexample is needed for the original class.

\begin{example}[Failure of Lifting for the Original Class]
\label{ex:original-lehmann}
Let $\Omega=\{\omega_1<\omega_2\}$ and consider
\[
 q=\frac16\begin{pmatrix}3&2&1\\1&2&3\end{pmatrix},
 \quad\quad
 p_1=\frac16\begin{pmatrix}3&3\\1&5\end{pmatrix},
 \quad\quad
 p_2=\frac16\begin{pmatrix}5&1\\3&3\end{pmatrix}.
\]
All entries are positive, and the likelihood ratios across states, read from left to right, are respectively
$(1/3,1,3)$, $(1/3,5/3)$, and $(3/5,3)$. Hence all three sources satisfy strict MLRP. Put $P=\{p_1,p_2\}$ and $Q=\{q\}$. Their posterior supports lie in the finite chains
\[
 C_P=\{(1-t,t)\mid t\in\{1/4,3/8,5/8,3/4\}\}.
\]
The benchmark chain is
\[
 C_Q=\{(1-t,t)\mid t\in\{1/4,1/2,3/4\}\}.
\]
These are compact subsets of $\operatorname{int}X$, so both sets are admissible. For every $D\in\sM_{\mathrm{Leh}}$,
\begin{equation}
            \max\{U(p_1,D),U(p_2,D)\}=U(q,D).
            \label{eq:original-primitive}
\end{equation}
Nevertheless, no tagged mixture $p_t=t p_1\oplus(1-t)p_2$ weakly outperforms $q$ throughout $\sM_{\mathrm{Leh}}$.
\end{example}

\begin{proof}
The source $p_1$ merges the second and third signals of $q$; $p_2$ merges its first and second signals. Thus $q$ Blackwell dominates each source, and the left side of \eqref{eq:original-primitive} cannot exceed its right side.

For the reverse inequality, a problem in $\sM_{\mathrm{Leh}}$ has at most two actions because $A=a^*(\Omega)$. A single action gives the same value under every source. With two actions, label them $a_1<a_2$. The loss matrix has the form
\[
       L=\begin{pmatrix}0&\ell_1\\\ell_2&0\end{pmatrix}.
\]
Here $\ell_1,\ell_2\geq0$. At a signal $s$ of $q$, choosing $a_2$ minimizes expected loss precisely when $\pi(\omega_1)\ell_1q(s\mid\omega_1)\leq\pi(\omega_2)\ell_2q(s\mid\omega_2)$. Because the likelihood ratio is increasing, ties can be resolved to obtain a deterministic threshold rule. Along the three ordered signals, its action indices are one of $(1,1,1)$, $(1,1,2)$, $(1,2,2)$, or $(2,2,2)$. The rule $(1,1,2)$ is implementable under $p_2$, and $(1,2,2)$ under $p_1$; constant rules are implementable under either. This argument also covers zero prior weights or zero losses by choosing a constant rule where appropriate. Consequently at least one source replicates an optimal rule of $q$, proving \eqref{eq:original-primitive} for every permitted loss and prior.

To rule out a single dominating tagged mixture, take the uniform prior and two problems with losses
\[
       L_1=\begin{pmatrix}0&1\\2&0\end{pmatrix}.
\]
The second loss matrix is
\[
       L_2=\begin{pmatrix}0&2\\1&0\end{pmatrix}.
\]
Both belong to the original class, with the lower action correct in state one and the higher action correct in state two. For $u_k=-L_k$, direct minimization of loss after each signal gives
\[
 \begin{array}{c|ccc}
       &U(p_1,D_k)&U(p_2,D_k)&U(q,D_k)\\ \hline
 D_1   &-5/12&-1/2&-5/12\\
 D_2   &-1/2&-5/12&-5/12
 \end{array}
\]
For completeness, the value for any two-state source $e$ is
\[
       U(e,D_k)=-\frac12\sum_s
        \min\{L_k(\omega_1,a_2)e(s\mid\omega_1),
               L_k(\omega_2,a_1)e(s\mid\omega_2)\}.
\]
The affinity of tagged mixtures now gives gaps relative to $q$ of $-(1-t)/12$ in $D_1$ and $-t/12$ in $D_2$. Matching $q$ in both problems would require $t=1$ and $t=0$ simultaneously.
\end{proof}

\begin{proposition}[No Universal Lifting for the Original Class]
\label{prop:original-no-lifting}
Even on the two-state space, no fixed binary relation between sources makes \eqref{eq:proposed-lifting} equivalent, for all finite admissible $P,Q$, to $U(P,D)\geq U(Q,D)$ for every $D\in\sM_{\mathrm{Leh}}$.
\end{proposition}

\begin{proof}
Singleton source sets would force the relation to coincide with pairwise value dominance on $\sM_{\mathrm{Leh}}$. Applied to Example~\ref{ex:original-lehmann}, the proposed equivalence would therefore imply a single tagged mixture dominating $q$ in both problems in the displayed value table, which is impossible. Every such tagged mixture has posterior support in $C_P$, so the singleton comparisons used in this argument remain admissible.
\end{proof}

The restriction $A=a^*(\Omega)$ matters for this counterexample: it is part of the original definition, and forces two actions in the two-state case. Enlarging the class to include additional actions that are never correct defines a different question. Likewise, bundling enlarges the class. The equal bundle of the two displayed problems yields value $-11/24$ from $P$ and $-5/12$ from $q$, reversing \eqref{eq:original-primitive}. The failure of primitive comparison to lift therefore persists under the literal original definition, although a different example establishes it.

\phantomsection

\end{document}